\documentclass[11pt]{article}
\usepackage{epsfig,amsfonts,latexsym,amsmath}
\usepackage{graphicx,color}
\usepackage{algorithm,algorithmic}
\usepackage{subfigure}
\newtheorem{theorem}{Theorem}[section]
\newtheorem{lemma}[theorem]{Lemma}
\newtheorem{corollary}[theorem]{Corollary}

\newtheorem{proposition}[theorem]{Proposition}

\newtheorem{definition}[theorem]{Definition}

\newenvironment{proof}{\noindent {\bf Proof:}}{$\hfill\square$}
\newenvironment{pf}[1]{\noindent{\bf Proof of #1:}}{$\hfill\square$}

\newcommand{\ceil}[1]{\left\lceil#1\right\rceil}
\newcommand\card[1]{\left| #1 \right|}
\newcommand\p{{\textsf{P}}}
\newcommand\np{{\textsf{NP}}}
\newcommand\dtime{{\textsf{DTIME}}}

\newcommand\set[1]{\left\{#1\right\}}
\newcommand{\pd}[1]{{\mathcal{P}(#1)}}
\newcommand{\npd}[2]{{\mathcal{P}^{#1}({#2})}}

\newcommand{\Optl}[2]{{\texttt{Opt}_{#1}(#2)}}
\newcommand{\T}{{\mathbf{T}}}
\newcommand{\Y}[3]{Y_{#1 \rightarrow #2}^{#3}}
\newcommand{\uh}[1]{\|#1\|}
\begin{document}
\title{Domination in graphs with bounded propagation: algorithms, formulations and hardness results}
\author{ Ashkan Aazami\\
\texttt{aaazami@uwaterloo.ca}\\
Department of Combinatorics and Optimization \\
University of Waterloo, Waterloo, ON   N2L 3G1, Canada \\
}
\date{\today}
\maketitle
\begin{abstract}
We introduce a hierarchy of problems between the \textsc{Dominating
Set} problem and the \textsc{Power Dominating Set} (PDS) problem
called the $\ell$-round power dominating set ($\ell$-round PDS, for
short) problem. For $\ell=1$, this is the \textsc{Dominating Set}
problem, and for $\ell\geq n-1$, this is the PDS problem; here $n$
denotes the number of nodes in the input graph. In PDS the goal is
to find a minimum size set of nodes $S$ that power dominates all the
nodes, where a node $v$ is power dominated if (1) $v$ is in $S$ or
it has a neighbor in $S$, or (2) $v$ has a neighbor $u$ such that
$u$ and all of its neighbors except $v$ are power dominated. Note
that rule (1) is the same as for the \textsc{Dominating Set}
problem, and that rule (2) is a type of propagation rule that
applies iteratively. The $\ell$-round PDS problem has the same set
of rules as PDS, except we apply rule (2) in ``parallel'' in at most
$\ell-1$ rounds. We prove that $\ell$-round PDS cannot be
approximated better than $2^{\log^{1-\epsilon}{n}}$ even for
$\ell=4$ in general graphs. We provide a dynamic
programming algorithm to solve $\ell$-round PDS optimally in
polynomial time on graphs of bounded tree-width. We present a PTAS
(polynomial time approximation scheme) for $\ell$-round PDS on
planar graphs for $\ell=O(\tfrac{\log{n}}{\log{\log{n}}})$. Finally,
we give integer programming formulations for $\ell$-round PDS.
\end{abstract}

\section{Introduction \label{Sec:Intro}}
The \textsc{Power Dominating Set} problem (PDS, for short) is a
covering problem in which the goal is to ``power dominate'' (cover)
all the nodes of a given undirected graph $G$ by picking as few
nodes as possible. There are two rules for power dominating the
nodes; the first one has a ``local'' effect but the second one
allows a type of indirect propagation.
More precisely, given a set of nodes $S$, the set of nodes that are
{\em power dominated} by $S$, denoted $\pd{S}$, is obtained as
follows
\begin{itemize}
\item[(R1)]
if node $v$ is in $S$, then $v$ and all of its neighbors are in
$\pd{S}$;
\item[(R2)] (propagation)
if node $v$ is in $\pd{S}$, one of its neighbors $w$ is not in
$\pd{S}$, and all other neighbors of $v$ are in $\pd{S}$, then $w$
is inserted into $\pd{S}$.
\end{itemize}
The PDS problem is to find a node-set $S$ of minimum size that power
dominates all the nodes (i.e., find $S\subseteq V$ with $\card{S}$
minimum such that $\pd{S}=V$). For example, consider the planar
graph $G=(V,E)$ in Figure \ref{Fig:Example}. The graph is obtained
by taking the union of $m$ paths on $k+1$ nodes that all meet at a
common node $v$; note that $\card{V}=k\cdot m+1$. It is easy to
check that the size of a minimum PDS in $G$ is $1$; by taking
$S=\set{v}$ we get $\pd{S}=V$. In more detail, by applying (R1) we
power dominate $v$ and all of its neighbors in the set $X_1$; after
that, each node in the set $X_1$ has exactly one neighbor that is
not power dominated yet, namely, its neighbor in $X_2$; thus, we can
sequentially apply propagation rule (R2) to each node in $X_1$ (in
any order) to power dominate all the nodes in $X_2$; continuing in
this way, all the nodes will be power dominated eventually.

The PDS problem arose in the context of monitoring electric power
networks. A power network contains a set of nodes and a set of edges
connecting the nodes. A power network also contains a set of
\emph{generators}, which supply power, and a set of \emph{loads},
where the power is directed to. In order to monitor a power network
we need to measure all the state variables of the network by placing
measurement devices. A Phasor Measurement Unit (PMU) is a
measurement device placed on a node that has the ability to measure
the voltage of the node and the current phase of the edges connected
to the node; PMUs are expensive devices. The goal is to install the
minimum number of PMUs such that the whole system can be monitored.
These units have the capability of monitoring remote elements via
propagation (as in Rule~2); see Brueni \cite{Bib:BR93}, Baldwin et
al.\ \cite{Bib:BAMBA93}, and Mili et al.\ \cite{Bib:MIBP91}. Most
measurement systems require one measurement device per node, but
this does not apply to PMUs; hence, PMUs give a huge advantage. To
see this in more detail consider a power network $G=(V,E)$, and
assume that the resistances of the edges in the power network are
known, and the goal is to measure the voltages of all nodes. For
simplicity, assume that there are no generators and loads. By
placing a PMU at node $v$ we can measure the voltage of $v$ and the
electrical current on each edge incident to $v$. Next, by using
Ohm's law we can compute the voltage of any node in the neighborhood
of $v$ (Rule~1). Now assume that the voltage on $v$ and all of its
neighbors except $w$ is known. By applying ohm's law we can compute
the current on the edges incident to $v$ except $\set{v,w}$. Next by
using Kirchhoff's law we compute the current on the edge
$\set{v,w}$. Finally, applying Ohm's law on the edge $\set{v,w}$
gives us the voltage of $w$ (Rule~2).

PMUs are used to monitor large system stability and to give warnings
of system-wide failures. PMUs have become increasingly popular for
monitoring power networks, and have been installed by several
electricity companies since 1988 \cite{Bib:MA02,Bib:BH99}. For
example, the USA Western States Coordinating Council (WSCC) had
installed around 34 PMUs by 1999 \cite{Bib:BH99}. By now, several
hundred PMUs have been installed world wide \cite{Bib:Phadke}. Some
researchers in electrical engineering regard PMUs as the most
important device for the future of power systems \cite{Bib:NU01}.

The PDS problem is \np-hard even when the input graph is bipartite
\cite{Bib:HAHHH02}. For further references for the PDS problem
please see
\cite{Bib:HAHHH02,Bib:KNMRR04,Bib:GUNR05,Bib:LILE05,Bib:BRHE05,Bib:AS06,Bib:DOHE06,Bib:ZHKA06}.
\begin{figure}[!h]
\begin{center}
    \input{lPDS.pstex_t}
\caption{} \label{Fig:Example}
\end{center}
\end{figure}
PDS is a generalization of the \textsc{Dominating Set} problem. A
\textsc{Dominating Set} of a graph $G=(V,E)$ is a set of nodes $S$
such that every node in the graph is either in $S$ or has a neighbor
in $S$. The problem of finding a \textsc{Dominating Set} of minimum
size in a given graph $G$ has been studied extensively in the past
20 years, see the books by Haynes et al.\
\cite{Bib:HAHESL98a,Bib:HAHESL98b}. The \textsc{Dominating Set}
problem is a well-known $\np$-hard problem \cite{Bib:GAJO79}. A
simple greedy algorithm achieves a logarithmic approximation
guarantee, \cite{Bib:JO74}, and modulo the $\p\not=\np$ conjecture,
no polynomial time  algorithm gives a better approximation
guarantee, \cite{Bib:LUYA94,Bib:FE98,Bib:RASA97}.

In this paper we introduce a hierarchy of problems between
\textsc{Dominating Set} and \textsc{PDS}, by adding a parameter
$\ell$ to PDS which restricts the number of ``parallel'' rounds of
propagation that can be applied. The rules are the same as PDS,
except we try to apply the propagation rule in parallel as much as
possible. In the first round we apply the rule (R1) to all the nodes
in $S$, and for the rest of the rounds we only consider ``parallel''
application of the propagation rule (R2). In every ``parallel''
round we power dominate all the new nodes that can be power
dominated by applying the propagation rule to all of the nodes that
are power dominated in the previous ``parallel'' rounds. Given a
parameter $\ell$, the $\ell$-round PDS  problem is the problem in
which we want to power dominate all of the nodes in at most $\ell$
parallel rounds. Clearly, the $\ell$-round PDS problem for $\ell=1$
is exactly the \textsc{Dominating Set} problem, and for a graph $G$
with $n$ nodes the $\ell$-round PDS problem for $\ell\geq n-1$ is
exactly the PDS problem. The notion of parallel propagation comes
from the fact that changes in the electrical network propagate in
parallel and not sequentially. A feasible solution for the PDS
problem provides a plan for installing monitoring devices to monitor
the whole power network, but it does not provide any guarantees on
the time-lag between a fault in the network and its detection.
Deducing information through a parallel round of propagation takes
one unit of time and in some applications we want to detect a
failure in the network after at most $\ell$ units of time. The
addition of the
 parameter $\ell$ achieves this time constraint.

The practical motivation for the $\ell$-round PDS problem has been
explained above. In addition, there are some theoretical
motivations. Although the PDS problem has been studied since 1993,
there are very few algorithmic results (including approximation
algorithms). The $\ell$-round PDS problem serves as a unified model
for studying the PDS problem and the \textsc{Dominating Set}
problem. The introduction of the parameter $\ell$ allows us to
examine the complexity of the problem in terms of $\ell$: how does
the threshold for the hardness of approximation vary in terms of
$\ell$ ? The hardness threshold is logarithmic for $\ell=1$ (the
\textsc{Dominating Set} problem) and it is
$\Omega(2^{\log^{1-\epsilon}{n}})$ for $\ell=n-1$ (the PDS problem).
Is the latter hardness threshold valid for constant $\ell$?
Moreover, for planar graphs there is extensive recent literature on
PTASs (polynomial time approximation schemes) for the
\textsc{Dominating Set} problem and its variants, but these results
do not apply to the PDS problem. A major open question in the area
is whether there exists a PTAS for the PDS problem on planar graphs.
One avenue that may lead to advances on this question is to design a
PTAS for the $\ell$-round PDS problem on planar graphs, for small
values of $\ell$. Integer programming formulations for the PDS
problem have been studied, and we give a new formulation in the last
part of this paper. Our formulation is based on the notion of
parallel propagation in $\ell$ rounds. We first give a formulation
for the $\ell$-round PDS problem, and then modify it to get another
formulation for the PDS problem.

\subsection{Our main results}
We initiate the study of a natural extension of the PDS problem and
prove the following main results.
\begin{itemize}
\item For general graphs, we present a reduction from the \textsc{MinRep} problem to the
$\ell$-round PDS problem which shows that $\ell$-round PDS for
$\ell\geq4$ cannot be approximated better than
$2^{\log^{1-\epsilon}{n}}$, unless
$\np\subseteq\dtime(n^{polylog(n)})$. We use a reduction similar to
one that has been used to prove the same hardness result for the PDS
and the directed PDS problems in a paper jointly authored with M.
Stilp \cite{Bib:AS06}.
\item We provide a dynamic programming algorithm to solve
the $\ell$-round PDS problem optimally in polynomial time on graphs
of bounded tree-width. This dynamic programming algorithm is based
on our reformulation for the $\ell$-round PDS problem. This
reformulation is an extension of the one introduced by Guo et al.\
\cite{Bib:GUNR05} for PDS. Guo et al.~\cite{Bib:GUNR05} gave a new
formulation for PDS in terms of ``valid orientation'' of the edges;
they use this formulation to design a dynamic programming algorithm
to solve the PDS problem optimally in linear time on graphs of
bounded tree-width.
\item We focus on planar graphs, and give a PTAS for $\ell$-round PDS
for $\ell=O(\frac{\log{n}}{\log{\log{n}}})$. Baker's PTAS
\cite{Bib:BA94} for the \textsc{Dominating Set} problem on planar
graphs is a special case of our result with $\ell=1$, and no similar
result of this type was previously known for $\ell>1$. We also show
that the $\ell$-round PDS problem in planar graphs is $\np$-hard for
all $\ell\geq 1$. Note that our PTAS does not apply to PDS in
general, because the running time is super-polynomial for
$\ell=\omega(\frac{\log{n}}{\log{\log{n}}})$.
\item Finally we study integer programming formulations for $\ell$-round
PDS.
\end{itemize}

Here is a brief discussion on the relation between this paper and
some previous joint work with M. Stilp on the PDS problem
\cite{Bib:AS06}. Although the results in \cite{Bib:AS06} do not
imply any of the results in this paper, there are two topics that
uses similar methods: (1) The hardness result for $\ell$-round PDS
(Theorem \ref{Thm:HardlDPDS}) extends the construction used to prove
the hardness of directed PDS in \cite{Bib:AS06}, (2) The dynamic
programming algorithm for $\ell$-round PDS (Section
\ref{Sec:Dynprog}) and the algorithm for directed PDS in
\cite{Bib:AS06} are both based on reformulation of the problems that
extend the methods of Guo et al.~\cite{Bib:GUNR05}.

\section{Preliminaries \label{Sec:Prelim}}
Most of the graphs that we consider here are undirected. Given a
graph $G=(V,E)$, we denote the number of nodes in the graph by $n$.
Sometimes the graphs that we consider have some directed edges in
addition to undirected edges. Given such a graph, we denote by
$d^-(v)$ and $d^+(v)$ the number of directed edges with $v$ as the
head and the tail respectively. Let $\widehat{G}$ denote the
underlying undirected graph obtained from $G$ by ignoring the
direction of the directed edges. Then the \emph{closed neighborhood}
of a node $u$ in $G$ is defined by $N[u]=\set{v: \set{u,v}\in
E(\widehat{G})}\cup \set{u}$. Now we define the ``parallel''
propagation rule formally.
\begin{definition}
Given a graph $G=(V,E)$ and a subset of nodes $S\subseteq V$, the
set of nodes that can be power dominated by applying at most $k$
rounds of parallel propagation, denoted by $\npd{k}{S}$, is defined
recursively as follows:
$$\npd{k}{S}=\left\{
\begin{array}{lr}
\bigcup_{v\in S}N[v]& k=1\\
\npd{k -1}{S}\bigcup\set{v:(u,v)\in E, N[u]\setminus\set{v}\subseteq \npd{k -1}{S}} & k\geq 2\\
\end{array}\right.$$
\end{definition}
We can now define the $\ell$-round PDS problem formally.
\begin{definition} \emph{($\ell$-round PDS)} \label{Def:lPDS}
Given a parameter $\ell$, the \emph{$\ell$-round PDS} problem is the
problem in which we are given a graph $G=(V,E)$ and the goal is to
find a minimum size subset of nodes $S\subseteq V$, such that
$\npd{\ell}{S}=V$.
\end{definition}
Given a graph $G$ and a parameter $\ell$, we denote by
$\Optl{\ell}{G}$ the size of the optimal solution for the
$\ell$-round PDS problem. It is easy to see that the size of an
optimal solution for $\ell$-round PDS does not increase by
increasing the value of the parameter $\ell$; to see this, consider
an optimal solution $S^*$ for $\ell$-round PDS and note that it is
also a feasible solution for $(\ell+1)$-round PDS. This proves the
following property of the $\ell$-round PDS problem.
\begin{proposition}\label{Pro:Monoton}
Let $G=(V,E)$ be a graph with $n$ nodes, then for any parameter
$1\leq \ell\leq\ell'\leq n$ we have $\Optl{\ell'}{G}\leq
\Optl{\ell}{G}$.
\end{proposition}

The $\ell$-round PDS problem is very different from the
\textsc{Dominating Set} problem, even for $\ell=O(1)$. The following
result shows that the size of an optimal solution for the
\textsc{Dominating Set} problem can be much bigger than the size of
an optimal solution for the $2$-round PDS problem.
\begin{proposition}
Let $\ell\geq 1$ be a given parameter, the ratio between the size of
the optimal solution for $\ell$-round PDS and $(\ell+1)$-round PDS
can be $\Theta({\frac{n}{\ell}})$ even on planar graphs.
\end{proposition}
\begin{proof}
Consider the graph $G$ that is obtained by taking $m$ paths, $P_1,
\ldots, P_m$, each of length $\ell+1$ that share a common node $v$
as shown in Figure \ref{Fig:Example} (with $k=\ell+1$). Note that
$G$ has $n=m\cdot(\ell+1)+1$ nodes. It is easy to check that $v$ can
power dominate the entire graph in exactly $\ell+1$ parallel rounds;
in the first round $v$ and $X_1$ are power dominated, in the second
round $X_2$ is power dominated, and so on. This shows that
$\Optl{\ell+1}{G}=1$. Next we prove that the size of the optimal
solution for $\ell$-round PDS is at least
$m=\tfrac{n-1}{\ell+1}=O(\frac{n}{\ell})$. For the sake of
contradiction assume that it is not, so there is an optimal solution
$S^*$ such that $\card{S^*}< m$. This implies that there is a path
$P_i$ that has no node in  $S^*$ except possibly the node $v$. This
means that the last node $u_i$ on the path $P_i$ (in $X_k$) cannot
be power dominated since there are only $\ell$ parallel rounds and
the distance between $v$ and $u_i$ is $\ell+1$. This contradicts the
fact that $S^*$ is a feasible solution for $\ell$-round PDS, so
$\Optl{\ell}{G}\geq m$. Also note that by taking $S$ to be the set
of middle nodes of the paths $P_1, \ldots, P_m$ we can power
dominate all of the nodes in $\ell$ rounds (in fact, in
$\frac{\ell}{2}$ rounds). This shows that $\Optl{\ell}{G}=m$, and
therefore the ratio between $\Optl{\ell}{G}$ and $\Optl{\ell+1}{G}$
is $m=O(\frac{n}{\ell})$.
\end{proof}

The $\ell$-round PDS problem can be generalized as follows. Given a
subset $V'$ of the nodes, find a minimum size set of nodes that
power dominates the nodes in $V'$ in at most $\ell$ parallel rounds.
We will use this generalization in our PTAS in Section
\ref{Sec:lPDS}.
\begin{definition} \emph{(Generalized $\ell$-round
PDS)}\label{Def:gen-lPDS} An instance of the \emph{generalized
$\ell$-round PDS} problem is given by a pair $\langle G,V'\rangle$,
where $G=(V,E)$ is an undirected graph and $V'\subseteq V$, and the
goal is to find a minimum size set of nodes $S$ such that
$V'\subseteq \npd{\ell}{S}$.
\end{definition}

\section{Dynamic programming for $\ell$-round PDS\label{Sec:Dynprog}}
In this section we provide a dynamic programming algorithm to solve
the generalized $\ell$-round PDS problem optimally. Our dynamic
programming algorithm is based on tree decompositions
\cite{Bib:DI00}. The running time of this algorithm is exponential
both in the tree-width of the given graph and the logarithm of the
parameter $\ell$. Therefore, we obtain a polynomial time algorithm
for graphs of bounded tree-width. Our dynamic programming algorithm
is based on our reformulation for the $\ell$-round PDS problem. This
reformulation is an extension of a new formulation for PDS
introduced by Guo et al.\ \cite{Bib:GUNR05}. Our dynamic programming
is similar to the dynamic programming for the \textsc{Dominating
Set} problem given in \cite{Bib:ABFKN02} and the dynamic programming
for PDS given in \cite{Bib:GUNR05}.
\begin{definition}
\label{Def:TWD} A tree decomposition of a graph $G=(V,E)$ is a pair
$\langle \set{X_i\subseteq V\ |\ i\in I},T\rangle$, where $T=(I,F)$
is a tree, satisfying the following properties:
\begin{enumerate}
\item $\bigcup_{i\in I}{X_i}=V;$
\item For every edge $\set{u,v}\in E$ there exists an $i\in I$ such
that $\set{u, v}\subseteq X_i;$
\item For all $i,j,k\in I$, if $j$ is on the unique path from $i$
to $k$ in $T$, then we have: $X_i\cap X_k\subseteq X_j.$
\end{enumerate}
The \emph{width} of $\langle \set{X_i\ |\ i\in I}, T\rangle$ is
defined as $max_{i\in I}{\card{X_i}}-1$. The \emph{tree-width} of
$G$ is defined as the minimum width over all tree decompositions.
The nodes of $T$ are called \emph{$T$-nodes}, and each $X_i$ is
called a \emph{bag}.
\end{definition}
For designing a dynamic programming algorithm based on tree
decomposition it is usually easier to work with \emph{nice tree
decompositions} \cite{Bib:KL94} that have a simple structure defined
as follows.
\begin{definition}
A tree decomposition $\langle \set{X_i| i\in I},T\rangle$ is called
a \emph{nice tree decomposition} if the following conditions are
satisfied:
\begin{enumerate}
\item $T$ is rooted, at $r$, and every node of $T$ has at most 2
children;
\item If a node $i\in I$ has two children $j$ and $k$, then $X_i=X_j=X_k$
(in this case $i$ is called a \textsc{Join Node});
\item If a node $i$ has one child $j$, then either of the following
holds:
\begin{enumerate}
\item $X_j\subset X_i$ and $\card{X_i\setminus X_j}=1$ (in this case $i$ is
called an \textsc{Insert Node}).
\item $X_i\subset X_j$ and $\card{X_j\setminus X_i}=1$  (in this case $i$ is called
a \textsc{Forget Node}).
\end{enumerate}
\end{enumerate}
\end{definition}
A tree decomposition can be transformed into a nice tree
decomposition \cite{Bib:KL94}. Therefore, to design a dynamic
programming algorithm based on tree decomposition, we can assume
that we are given a nice tree decomposition of the input graph.
\begin{lemma}[\cite{Bib:KL94}] \label{Lem:Nice}Given a tree decomposition of a graph $G$ that has
width $k$ and $O(n)$ nodes, where $n$ is the number of nodes of $G$,
a nice tree decomposition of $G$ that also has width $k$ and $O(n)$
nodes can be found in time $O(n)$.
\end{lemma}

Now we provide a reformulation, called \emph{timed-orientation}, for
(generalized) $\ell$-round PDS in terms of the orientation of the
edges and labeling of the nodes. This reformulation makes it
possible to design a dynamic programming algorithm to solve the
(generalized) $\ell$-round PDS problem optimally in polynomial time
on graphs of bounded tree-width. In this formulation we orient the
edges in order to show the direction that the propagation rule is
applied; if the node $w$ is power dominated by applying the
propagation rule on $v$ then we orient the edge $\set{v, w}$ from
$v$ toward $w$. Moreover, if $w$ is not power dominated through $v$
we will leave the edge $\set{v, w}$ undirected. In order to keep
track of the round in which every node is power dominated, we
introduce a time vector $\set{t_v: v\in V}$. The round in which $v$
is power dominated is denoted by $t_v$ and it can take any value
from the set $\set{0, \ldots, \ell}\cup\set{+\infty}$; the $+\infty$
is used to denote a node that will not be power dominated (this is
needed for generalized $\ell$-round PDS). Consider a directed edge
$(u,v)$. There are two cases: either $t_v=1$ or $t_v>1$. In the
former case $u$ should be in the (optimal) solution, since $t_v=1$
means $v$ is power dominated at the first round by applying the
domination rule. In the latter case $u$ and all of its neighbors
except $v$ (i.e\ $N[u]\setminus v$) should be power dominated before
we can apply the propagation rule. In fact, $v$ will be power
dominated right after the round where the last node in
$N[u]\setminus v$ is power dominated (see the property P5 in the
following definition).

\begin{definition} \emph{(valid timed-orientation)}
\label{Def:ValTOR} Let $\langle G,V'\rangle$ be an instance of the
generalized $\ell$-round PDS problem where $G=(V,E)$ is a graph and
$V'\subseteq V$ is a subset of nodes. A \emph{valid
timed-orientation} for $\langle G,V'\rangle$ is a graph
$G_o=(V,E_d\cup E_u)$ such that for every $\set{u,v}\in E$ either
there is a directed edge $(u,v)$ or $(v,u)$ in $E_d$ or an
undirected edge $\set{u,v}$ in $E_u$, together with the time vector
$\set{t_v:v\in V}$ (possible values for $t_v$ are $\set{0, 1,
\ldots, \ell}\cup\set{+\infty}$) that satisfies the following
properties:
\begin{itemize}
\item[\emph{(P1)}] $\forall v\in V': 0\leq t_v\leq \ell$,
\item[\emph{(P2)}] $\forall v\in V: 1\leq t_v\leq \ell \Rightarrow d^-(v)=1$,
\item[\emph{(P3)}] $\forall v\in V: t_v=+\infty\Rightarrow d^-(v)=d^+(v)=0$,
\item[\emph{(P4)}] $\forall v\in V: t_v=0\Rightarrow d^-(v)=0$,
\item[\emph{(P5)}] $\forall (u,v)\in E_d: t_v=\left\{
\begin{array}{lr}
1 & \mbox{if } t_u=0\\
1+\max\set{t_w: w\in N[u]-v}& \emph{otherwise}.\\
\end{array}\right.$
\end{itemize}
The set $O=\set{v\in V: t_v=0}$ is called the \emph{origin} of the
valid timed-orientation.
\end{definition}
Now we show that the existence of a valid timed-orientation with $S$
as the origin is equivalent to having $S$ as a feasible solution to
the $\ell$-round PDS problem.
\begin{lemma}
Let $G=(V,E)$ be an undirected graph, and $\langle G,V'\rangle$ be
an instance of the generalized $\ell$-round PDS problem. Then
$S\subseteq V$ is the origin of a valid timed-orientation if and
only if $V'\subseteq \npd{\ell}{S}$.
\end{lemma}
\begin{proof}
Suppose $S\subseteq V$ power dominates $V'$ in at most $\ell$
parallel rounds. Now we construct a valid timed-orientation with $S$
as the origin. We orient the edges in the same way as the
propagation rule applies. Consider an edge $\set{u, v}$ and assume
that $v$ is power dominated by applying power domination rules to
$u$. Then we orient the edge $\set{u, v}$ from $u$ toward $v$. Let
$E_d$ be the set of oriented edges and let $E_u$ be the rest of the
edges. This defines the graph $G_o=(V,E_d\cup E_u)$. Now we define
the time vector. First, for any node $v$ that is not power
dominated, set $t_v=+\infty$. Second, for every $v\in S$, set
$t_v=0$. Finally, define $t_v$ for the rest of the nodes as
$t_v=\min{\set{r\geq 1: v\in\npd{r}{S}}}$. It is straightforward to
check that the above orientation and the time vector satisfy all of
the properties P1 to P5 in the definition of the valid
timed-orientation. Therefore, there is a valid timed-orientation
with $S$ as its origin.

Now assume that $G$ has a valid timed-orientation with $S$ as the
origin. So there is an oriented graph $G_o=(V,E_d\cup E_u)$ and a
time vector $\set{t_v:v\in V}$ that satisfy the properties P1 to P5
of the valid timed-orientation. Define $V_r=\set{v\in V:0\leq
t_v\leq r}$. We now prove by induction that $V_r\subseteq\npd{r}{S}$
which implies that $V_{\ell}\subseteq\npd{\ell}{S}$. First note that
this proves the lemma, since any node $v\in V'$ has $0\leq
t_v\leq\ell$ by the property P1. Note that any node $v$ with $t_v=0$
is in $S$, so it is immediately power dominated. Also any node $v$
with $t_v=1$ is power dominated (by the domination rule) in the
first round, since by P2 it has an incoming directed edge, say
$(w,v)$, and by P5 the node $w$ is in S. Hence, $V_1\subseteq
\npd{1}{S}$. Now assume that the induction hypothesis holds for
$r=k$ (where $k\geq 1$), that is $V_k\subseteq \npd{k}{S}$.  We will
show that it also holds for $r=k+1$. Consider a node $v$ such that
$t_v=k+1$. By P2 it has exactly one incoming edge, say $(u,v)\in
E_d$. The node $u$ has $t_u\geq 1$ (since if $t_u=0$ then P5 implies
$t_v=1$), so by the property P5 any $w\in N[u]-v$ has $t_w\leq k$.
Therefore, $N[u]\setminus\set{v}\subseteq \npd{k}{S}$. This is
correct for any $v$ with $t_v=k+1$, so $X=\set{v\in V:
t_v=k+1}\subseteq\set{v:(u,v)\in E,
N[u]\setminus\set{v}\subseteq\npd{k}{S}}$. Therefore, by the
definition of the parallel propagation  $X$ can be power dominated
in one parallel round, so we have $V_{k+1}=V_k\cup
X\subseteq\npd{k+1}{S}$. This proves the induction step and
completes the proof.
\end{proof}

\begin{theorem}
Given a pair $\langle G,V'\rangle$ where $G=(V,E)$ is a graph with
tree-width $k$ and $V'\subseteq V$, a minimum size set $S\subseteq
V$ such that $V'\subseteq \npd{\ell}{S}$ can be obtained in time
$O(c^{m_e+k\log{\ell}}\cdot \card{V})$, for some global constant
$c$, where $m_e$ is the maximum number of edges that a bag can have.
\label{Thm:LGPDSDyn}
\end{theorem}
\begin{proof}
Please refer to Appendix \ref{Apx:DynProg} for the details of the
dynamic programming.
\end{proof}

In general graphs the number of edges in a bag is at most $\left(\begin{smallmatrix}k+1\\
2\end{smallmatrix}\right)$ but in planar graphs, it is linear in the
number of nodes, where $k$ denotes the tree-width. Therefore, the
above theorem implies the following corollary.
\begin{corollary} \label{Cor:LGPDSSDyn}
Given a pair $\langle G,V'\rangle$ where $G=(V,E)$ is a planar graph
with tree-width $k$ and $V'\subseteq V$, a minimum size set
$S\subseteq V$ such that $V'\subseteq \npd{\ell}{S}$ can be obtained
in time $O(c^{k\log{\ell}}\cdot \card{V})$, for some global constant
$c$.
\end{corollary}

\section{$\ell$-round PDS\label{Sec:lPDS} on planar graphs}
In this section we present a PTAS (polynomial time approximation
scheme\footnote{A \emph{polynomial time approximation scheme} (PTAS)
is an algorithm that given any fixed $\epsilon>0$ provides a
solution with cost within $(1+\epsilon)$ times the optimal value in
polynomial time.}) for the $\ell$-round PDS problem on planar graphs
when $\ell=O(\frac{\log{n}}{\log{\log{n}}})$. Baker's PTAS
\cite{Bib:BA94} for the \textsc{Dominating Set} problem on planar
graphs is a special case of our result with $\ell=1$, but there are
no previous results of this type for $\ell>1$. Our PTAS works in the
same fashion as Baker's PTAS, but our analysis and proofs are novel
contributions of this paper. Demaine and Hajiaghayi
\cite{Bib:DEHA05} recently used bidimensionality theory to design
PTASs for some variants of the \textsc{Dominating Set} problem on
planar graphs (such as \textsc{Connected Dominating Set}), but their
methods do not apply to the $\ell$-round PDS problem because the
relevant parameter is not bidimensional. We also have the following
$\np$-hardness result for $\ell$-round PDS on planar graphs.
\begin{proposition}\label{Thm:NP-lPDS} For any $\ell\geq 1$
the $\ell$-round PDS problem is an $\np$-hard problem even on planar
graphs.
\end{proposition}
\begin{proof}
We use a modification of the reduction that has been used to prove
that the PDS problem on planar graphs is $\np$-hard
\cite{Bib:KNMRR04,Bib:GUNR05}. Please refer to Appendix
\ref{Apx:Proof} for more details.
\end{proof}

Now we describe our PTAS for the $\ell$-round PDS problem on planar
graphs, when $\ell$ is small. First we provide some useful
definitions and notations. Consider an embedding of a planar graph
$G$. We define the nodes at level $i$ denoted by $L_i$ as follows
\cite{Bib:BA94}. Let $L_1$ be the set of nodes on the exterior face
in the given embedding of $G$. For $i>1$, the set $L_i$ is defined
as the set of nodes on the exterior face of the graph induced on
$V\setminus \cup_{j=1}^{i-1}{L_j}$. We denote by
$L(a,b)=\cup_{i=a}^{b}{L_i}$ the set of nodes at levels $a$ through
$b$. A planar graph is called \emph{$k$-outerplanar} if it has an
embedding where no node is at level greater than $k$. For example,
consider the graph in Figure \ref{Fig:Bad-Part}. Clearly, the graph
is a $2$-outerplanar graph. The set $L_1=\set{u_1, u_2, \ldots,
u_8}$ is the set of nodes at level $1$ and the set $L_2=\set{v_1,
v_2, \ldots, v_8}$ is the set of nodes at level $2$. Given a graph
$G=(V,E)$ and $V'\subseteq V$, we denote the subgraph induced on
$V'$ by $G[V']$.

Before describing our PTAS, let us look at Baker's PTAS
\cite{Bib:BA94} for the \textsc{Dominating Set} problem on planar
graphs. Given the parameter $\epsilon=\frac{1}{k}$, Baker's
algorithm finds a feasible solution with size within $(1+\epsilon)$
 times the optimal value. This algorithm considers $k$
different decompositions $\mathcal{D}_1, \ldots, \mathcal{D}_k$ of
the nodes of $G$ and then finds a feasible solution for each of
them.  The $i\mbox{th}$ decomposition consists of blocks with $k+1$
consecutive levels. The $j\mbox{th}$ block in decomposition
$\mathcal{D}_i$ contains nodes of levels $j k+i$ through $(j+1) k+i$
(note that each $\mathcal{D}_i$ is obtained from $\mathcal{D}_1$ by
shifting the levels). Also note that every two consecutive blocks in
a given decomposition share a common level. Next, the algorithm
solves the \textsc{Dominating Set} problem optimally on each block
of $\mathcal{D}_i$. This is possible since each block is
$(k+1)$-outerplanar; hence, it has tree-width at most $3(k+1)-1$ and
the \textsc{Dominating Set} problem can be solved optimally by
dynamic programming \cite{Bib:BA94} (the result in Section
\ref{Sec:Dynprog} for $\ell=1$ also shows this fact). Then, it takes
the union of the optimal solutions for the blocks in $\mathcal{D}_i$
to obtain a feasible solution for the \textsc{Dominating Set}
problem in the graph $G$. Let $S_i$ denote this feasible solution.
Finally, the algorithm outputs the solution that has minimum size,
i.e. $\min_{i=1,\ldots,k}{\card{S_i}}$, among all $k$
decompositions. It is not hard to argue that the size of this
solution is within $\tfrac{k+1}{k}=1+\epsilon$ times the optimal
value. The key property that is needed for this argument is the fact
that consecutive blocks share a common level.
\begin{figure}[htbp]
\begin{center}
\input{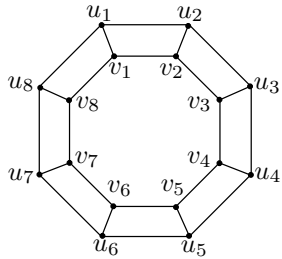}
\caption{A $2$-outerplanar graph} \label{Fig:Bad-Part}
\end{center}
\end{figure}

Now let us describe our PTAS informally. Consider a parameter $k$,
that is a function of the parameter $\ell$ and the approximation
factor ($1+\epsilon$); $k$ will be defined later in the formal
description of our algorithm. We decompose the graph in $k$
different ways, $\mathcal{D}_1, \ldots,\mathcal{D}_k$. In each
decomposition the graph is decomposed into blocks of $k+4\ell-2$
consecutive levels. The $j\mbox{th}$ block in $\mathcal{D}_i$ is
defined as $B_{i,j}=L(jk+i-2\ell+1,(j+1)k+i-1+2\ell-1)$. We denote
the $k$ middle levels of $B_{i,j}$ by $C_{i,j}=L(jk+i,(j+1)k+i-1)$.
In our PTAS, for each decomposition $\mathcal{D}_i$, we optimally
solve instances of the generalized $\ell$-round PDS problem
($\mathcal{I}_{i,j}=\langle G[B_{i,j}],C_{i,j}\rangle$) defined for
each block in $\mathcal{D}_i$. Note that each instance
$\mathcal{I}_{i,j}$ can be optimally solved by using the dynamic
programming algorithm given in Section \ref{Sec:Dynprog}. Let
$\mathcal{O}_{i,j}$ denote the optimal solution for this instance.
Then we take the union of the solutions corresponding to blocks in
$\mathcal{D}_i$, $\Pi_i=\cup_{j\geq 0}{\mathcal{O}_{i,j}}$. By doing
this for all $k$ decompositions we get $k$ different feasible
solutions for the original graph $G$. Finally, we choose the
solution with minimum size among these $k$ solutions. We will see
that the $2\ell-1$ extra levels around the $k$ middle levels and the
common levels between consecutive blocks plays an important role in
the feasibility and the near optimality of the final output of the
algorithm.

To make the role of common levels clear, consider an instance of
$4$-round PDS shown in Figure \ref{Fig:Bad-Part}. Assume that we
partition the planar graph into two levels. The first level is the
outer cycle and the second level is the inner cycle. It is easy to
check that the size of an optimal solution in any one of these
partitions is $1$. For example $\set{u_1}$ and $\set{v_5}$ are
optimal solutions for partition $1$ and partition $2$ respectively.
But if we consider the original graph $G$, it is straightforward to
check that $S=\set{u_1}\cup\set{v_5}$ is not a feasible solution for
the instance $G$ of $4$-round PDS; the set of nodes that can be
power dominated in at most four parallel rounds is
$\npd{4}{S}=\set{u_1, u_2, u_5, u_8, v_1, v_4, v_5, v_6}$.  Note
that the propagation rule was applied in the subgraph but not in the
original graph, so in the subgraph $u$ may have all but one node of
$N[u]$ in $\npd{i}{S}$ but this need not hold for the original graph
(see $u_8$, for example).

To prevent such a problem, we need to consider extra levels around
each block as we did in the blocks $B_{i,j}$ above. We will show
that a feasible solution, found by our algorithm, for the instance
$\mathcal{I}_{i,j}$ will power dominate at least its $k$ middle
levels in the original graph $G$. This implies that the union of the
solutions for the blocks in a given decomposition will be a feasible
solution for $G$, since the union of the $k$ middle levels of the
blocks covers all the nodes in $G$.

\begin{algorithm}[h]
\caption{PTAS for $\ell$-round PDS}
 \label{Alg:PTAS}
\begin{algorithmic}[1]
  \STATE Given a planar embedding of the graph $G$, and the
  parameter $0<\epsilon \leq 1$.
  \STATE Let $k=4\cdot\lceil\frac{\ell}{\epsilon}\rceil$.
  \FOR{$i=1$ to $k$}
  \FORALL{$j\geq 0$}
    \STATE Solve ``generalized'' $\ell$-round PDS on $\langle G[B_{i,j}],C_{i,j}\rangle$
    \STATE Let $\mathcal{O}_{i,j}$ be an optimal solution for $\langle
    G[B_{i,j}],C_{i,j}\rangle$
  \ENDFOR
  \STATE $\Pi_i=\cup_{j\geq 0}{\mathcal{O}_{i,j}}$
  \ENDFOR
  \STATE $r\leftarrow argmin\set{\card{\Pi_i}: i=1, \cdots, k}$
  \STATE Output $\Pi_O=\Pi_r$.
\end{algorithmic}
\end{algorithm}
\begin{theorem}\label{Thm:LPDS}
Let $\ell$ be a given parameter, where
$\ell=O(\frac{\log{n}}{\log{\log{n}}})$. Then Algorithm
\ref{Alg:PTAS} is a PTAS for the $\ell$-round PDS problem on planar
graphs.
\end{theorem}
\begin{proof}
Let $\Pi^*$ be an optimal solution for $\ell$-round PDS in $G$. To
prove the theorem, it is enough to prove the following two claims:
1) $\mathcal{O}_{i,j}$ is a feasible solution for the instance
$\langle G,C_{i,j}\rangle$ of the generalized $\ell$-round PDS
problem, 2) $\Pi^* \cap B_{i,j}$ is a feasible solution for $\langle
G[B_{i,j}],C_{i,j}\rangle$. First let us see how the theorem follows
from the above claims. The first claim shows that $\Pi_i$, for each
$i$, is a feasible solution for the $\ell$-round PDS problem for
$G$, since $\cup_{j\geq 0}{C_{i,j}}=V$.
The second claim shows that $\card{\mathcal{O}_{i,j}}\leq
\card{\Pi^* \cap B_{i,j}}$, so
$\card{\Pi_i}\leq\sum_{j}{\card{\Pi^*\cap B_{i,j}}}$. In the right
hand side we counted the nodes in the optimal solution twice on
$4\ell-2$ common levels between any two consecutive blocks. By
considering all values of parameter $i$, $1\leq i\leq k$, we can
find an $i$ such that the number of double counted nodes in the
optimal solution $\Pi^*$ is at most $\frac{4\ell-2}{k}\card{\Pi^*}$.
This implies that $\card{\Pi_O}\leq
(1+\frac{4\ell-2}{k})\card{\Pi^*}$, so by setting
$k=4\cdot\lceil\frac{\ell}{\epsilon}\rceil$ we get a
$(1+\epsilon)$-approximation algorithm.

Now we analyze the running time of the algorithm. In step (5) we
solve an instance of the generalized $\ell$-round PDS problem. The
graph in this instance has $k+4\ell-2$ levels, so it is a
$(k+4\ell-2)$-outerplanar graph. It is known that the tree-width of
any $d$-outerplanar graph is at most $3d-1$ \cite{Bib:BO88}.
Therefore step (5) of our algorithm can be done in
$c^{O({\ell/{\epsilon}\log{\ell}})}$ time by Corollary
\ref{Cor:LGPDSSDyn}, since the graph has tree-width at most
$3(k+4\ell)\leq 12(\lceil\frac{\ell}{\epsilon}\rceil+\ell)$. This
shows that $\ell$ should be $O(\frac{\log{n}}{\log{\log{n}}})$ in
order to have a polynomial time algorithm and in this case step (5)
can be done in $n^{O(\frac{1}{\epsilon})}$ time. Also note that the
value of $j$ can be at most $\frac{n}{k}$, since the number of
levels of a planar graph is at most $n$ (the number of nodes) and
each $C_{i,j}$ has $k$ levels. This shows that the algorithm
executes step (5) at most $k\times \tfrac{n}{k}= n$ times. Also the
steps (8) and (10) can be done in polynomial time. Therefore the
running time of the algorithm is polynomial time in the number of
nodes of $G$ for a fixed $\epsilon$.

Now we prove the two claims stated above: 1) $\mathcal{O}_{i,j}$ is
a feasible solution for $\langle G,C_{i,j}\rangle$, 2) $\Pi^* \cap
B_{i,j}$ is a feasible solution for $\langle
G[B_{i,j}],C_{i,j}\rangle$.

\noindent{\bf Proof of the first claim:} We know that
$\mathcal{O}_{i,j}$ is a feasible solution for $\langle
G[B_{i,j}],C_{i,j}\rangle$. Let $t_v$ denote the round in which $v$
was power dominated in $G[B_{i,j}]$. So any node $v\in C_{i,j}$ and
possibly some nodes $v\in B_{i,j}\setminus C_{i,j}$ satisfy $0\leq
t_v\leq \ell$. For simplicity we use $L^{s}$ to denote the levels
$L(jk+i-2(\ell-s),(j+1)k+i-1+2(\ell-s))$ for any $s\geq 1$, and also
$L^{0}=B_{i,j}$. Observe that $L^{s+1}$ (for $s\geq 1$) is obtained
from $L^{s}$ by deleting the first two levels of $L^{s}$ and the
last two levels of $L^{s}$.

First note that by taking $\mathcal{O}_{i,j}$ all nodes $v$ with
$t_v=0$ in $L^{0}$ are power dominated in the graph $G$. Now we
claim that the following statement is correct: for each $s$, $1\leq
s \leq\ell$, all nodes $v$ with $t_v\leq s$ in $L^{s}$ are power
dominated in the graph $G$. We prove the statement by induction on
$s$. The base case $s=1$ is trivial by applying the first rule of
PDS, since any node $v\in L^{1}$ with $t_v=1$ had a neighbor $u\in
L^0$ with $t_u=0$. Therefore all nodes with $t_v\leq 1$ in $L^{1}$
are power dominated. Assume that the statement is correct for all
$s< s'$. Consider a node $v$ with $t_v=s'$ which lies in $L^{s'}$,
and assume that it was power dominated by applying the propagation
rule on $u$. It is easy to see that $u$ is inside the levels
$L(jk+i-2(\ell-s')-1,(j+1)k+i-1+2(\ell-s')+1)$, and therefore all
neighbors of $u$ are inside the levels
$L(jk+i-2(\ell-s')-2,(j+1)k+i-1+2(\ell-s')+2)=L^{s'-1}$  (see Figure
\ref{Fig:Layer-Proof}). As $v$ was power dominated by $u$, any $w\in
N[u]-v$ satisfies $t_w\leq s'-1$. By the induction hypothesis any
node $w\in N[u]-v$ is already power dominated in the graph $G$. The
propagation rule can be applied to $u$ to power dominate $v$. This
completes the induction step and proves the statement. An important
property that should be preserved is that all nodes $v$ with
$t_v=s'$ should be power dominated in parallel. It is easy to see
that in the above proof any such node can be power dominated in
parallel since they were power dominated in parallel in the induced
graph $G[B_{i,j}]$.

To prove the first claim it is enough to note that the set
$L^{\ell}$ is exactly $C_{i,j}$, so all the nodes in $C_{i,j}$ with
$0\leq t_v\leq\ell$ (which is exactly all the nodes in $C_{i,j}$)
can be power dominated in the graph $G$.
\begin{figure}[htbp]
\begin{center}
\input{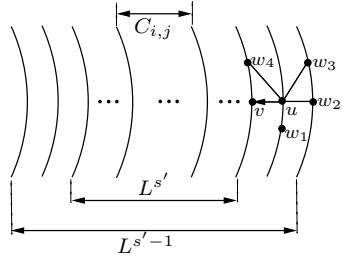}
\caption{Induction step} \label{Fig:Layer-Proof}
\end{center}
\end{figure}

\noindent{\bf Proof of the second claim:} We know that $\Pi^*$ is a
feasible solution for $G$. Let $t_v$ denote the round in which node
$v\in B_{i,j}$ was power dominated in $G$. So each node $v\in
B_{i,j}$ satisfies: $0\leq t_v\leq \ell$. Define $L^s$ as before.
The same induction hypothesis as above will prove the statement: for
each $s$, $0\leq s \leq\ell$, all nodes $v$ with $t_v\leq s$ in
$L^{s}$ can be power dominated in the induced subgraph $G[B_{i,j}]$.
The proof is similar to the first claim. Note that the set
$L^{\ell}$ is exactly $C_{i,j}$, so all the nodes in $C_{i,j}$ can
be power dominated in $G[B_{i,j}]$ in at most $\ell$ parallel
rounds.
\end{proof}

\section{Hardness of $\ell$-round PDS}
In this section we prove the following result by a reduction from
the \textsc{Minrep} problem. The reduction given here is similar to
the reduction used to prove the same hardness of approximation for
the Directed PDS problem \cite{Bib:AS06}.
\begin{theorem}\label{Thm:HardlDPDS}
The $\ell$-round PDS problem for any $\ell\geq 4$ cannot be
approximated within $2^{\log^{1-\epsilon}{n}}$ ratio, for any fixed
$\epsilon > 0$, unless $\np\subseteq \dtime(n^{polylog(n)})$.
\end{theorem}

In the \textsc{MinRep} problem we are given a bipartite graph
$G=(A,B,E)$ with a partitioning of $A$ and $B$ into equal size
subsets, say $A=\bigcup_{i=1}^{q_A}A_i$ and
$B=\bigcup_{i=1}^{q_B}B_i$, where
$\card{A_i}=m_A=\frac{\card{A}}{q_A}$ and
$\card{B_i}=m_B=\frac{\card{B}}{q_B}$. This partitioning naturally
defines a super bipartite graph
${\cal{H}}=({\cal{A}},{\cal{B}},{\cal{E}})$. The super nodes of
${\cal H}$ are ${\cal{A}}=\{A_1,A_2,\cdots,A_{q_A}\}$ and
${\cal{B}}=\{B_1,B_2,\cdots,B_{q_B}\}$, and the super edges are
${\cal E}=\{A_iB_j| \exists a\in A_i, b\in B_j: \set{a,b}\in
E(G)\}$. We say that a super edge $A_iB_j$ is covered by
$\set{a,b}\in E(G)$ if $a\in A_i$ and $b\in B_j$. The goal in
\textsc{MinRep} is to pick the minimum number of nodes $A'\cup
B'\subseteq V(G)$ from $G$ to cover all the super edges in
${\cal{H}}$. The following theorem states the hardness of the
\textsc{MinRep} problem \cite{Bib:KKL04}.
\begin{theorem} \cite{Bib:KKL04}
The \textsc{MinRep} problem cannot be approximated within the ratio
$2^{{log}^{1-\epsilon}{n}}$, for any fixed $\epsilon >0$, unless
$\np\subseteq \dtime(n^{polylog(n)})$, where $n=\card{V(G)}$.
\end{theorem}

\noindent{\bf The reduction:} Theorem \ref{Thm:HardlDPDS} is proved
by a reduction from the \textsc{MinRep} problem. In the following we
create an instance $\overline{G}=(\overline{V},\overline{E})$ of
$\ell$-round PDS from a given instance $G=(A,B,E)
({\cal{H}}=({\cal{A}},{\cal{B}},{\cal{E}})$)  of the \textsc{MinRep}
problem.
\begin{enumerate}
\item Add a new node $w^*$ (master node) to the graph $G$, and add
an edge between $w^*$ and all the other nodes in $G$. Also add three
new nodes $w^*_1, w^*_2, w^*_3$ and connect them to $w^*$.
\item $\forall i\in\set{1,\ldots,q_A}, j\in\set{1,\ldots,q_B}$ do
the following:
\begin{enumerate}
\item Let $E_{ij}=\set{e_1,e_2,\ldots,e_{\kappa}}$ be the set of edges between $A_i=\{a_{i_1}, \ldots, a_{i_{m_A}}\}$ and
$B_j=\{b_{j_1}, \ldots, b_{j_{m_B}}\}$ in $G$, where $\kappa$ is the
number of edges between $A_i$ and $B_j$.
\item Remove $E_{ij}$ from $G$.
\item Let the edge $e_q\in E_{i,j}$ be incident to $a_{i_q}$ and $b_{j_q}$ (in $G$). In this labeling for simplicity the same node might get different
labels. Let $D_{ij}$ be the graph in Figure \ref{Fig:HardlDPDS} (a
dashed line shows an edge between a node and the master node $w^*$).
\begin{figure}[!t]
\begin{center}
\input{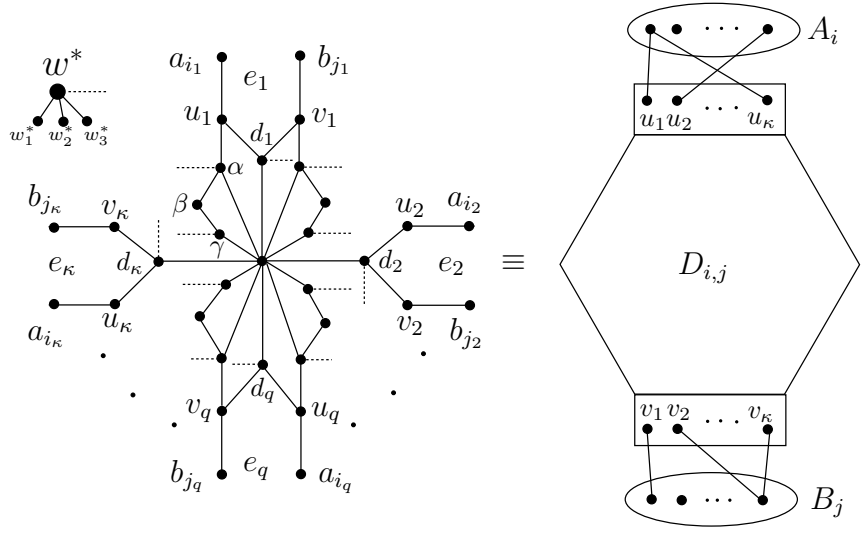}
\caption{The graph $D_{ij}$} \label{Fig:HardlDPDS}
\end{center}
\end{figure}
 Make $\lambda=4$ new copies of the graph $D_{ij}$ and then
identify nodes $a_{i_q}$'s, $b_{j_q}$'s with the corresponding nodes
in $A_i$ and $B_j$ (in $G$). Note that the $\lambda$ copies are
sharing the same set of nodes, $A_i$ and $B_j$, but other nodes are
disjoint.
\end{enumerate}
\item Let $\overline{G}=(\overline{V},\overline{E})$ be the obtained
graph (See Figure \ref{Fig:HardConst} for an illustration).
\end{enumerate}
\begin{figure}[!hb]
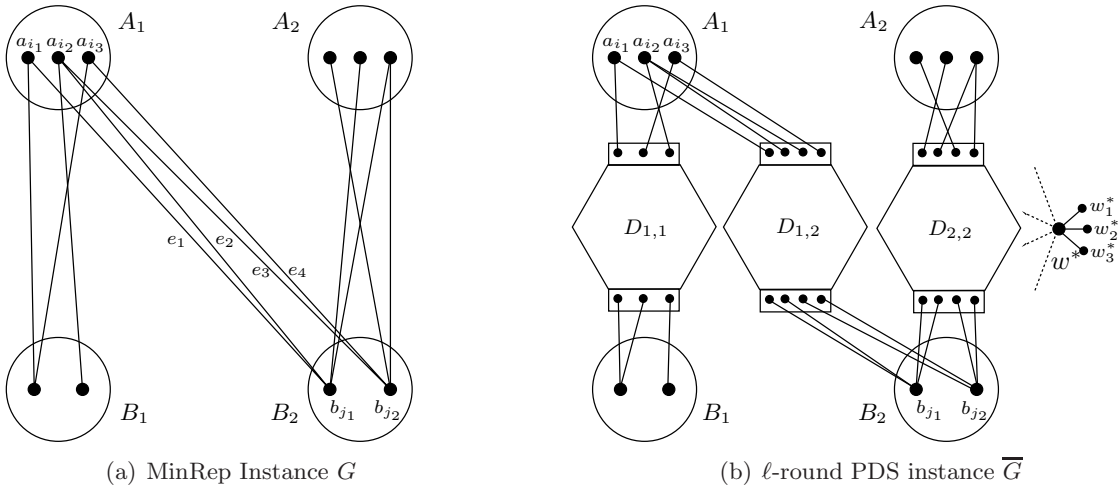

\begin{center}
\subfigure[MinRep Instance $G$] {
    \input{HardConsta.pstex_t}
} \hspace{1cm} \subfigure[$\ell$-round PDS instance $\overline{G}$]
{
    \input{HardConstb.pstex_t}
} \caption{The hardness construction} \label{Fig:HardConst}
\end{center}
\end{figure}
\noindent{\bf The analysis:} The next lemma shows that the size of
an optimal solution in $\ell$-round PDS is exactly one more than the
size of an optimal solution in the \textsc{MinRep} instance. The
number of nodes in the constructed graph is at most
$\card{V(\overline{G})}\leq 4+\card{V(G)}+10\lambda\card{E(G)}$.
This shows that the above reduction is a \emph{gap preserving
reduction} from \textsc{MinRep} to $\ell$-round PDS with the same
gap (hardness ratio) as the \textsc{MinRep} problem. Therefore the
following lemma will complete the proof of the above theorem. As we
mentioned above, the reduction given here is similar to the one used
for proving the hardness of the directed PDS problem
\cite{Bib:AS06}. One important part of the above construction (see
Figure \ref{Fig:HardlDPDS}) is the gadget on the set of nodes
$\set{\alpha, \beta, \gamma}$. Note that there should be such a
gadget between the center node in $D_{i,j}$ and each $u_q$ and
$v_q$; in Figure \ref{Fig:HardlDPDS} not all of the gadgets are
shown (for example between $u_2$ and the center node). This gadget
introduces direction into undirected construction, it allows the
propagation in only one direction. After the center node of
$D_{i,j}$ is power dominated, all of the other nodes in $D_{i,j}$
get power dominated (by the propagation rule). On the other hand,
the power domination cannot propagate through the gadget in the
other direction (toward the center node).
\begin{lemma}
The pair $(A^*,B^*)$ is an optimal solution to the instance
$G=(A,B,E)$ of the \textsc{MinRep} problem if and only if
$\Pi^*=A^*\cup B^*\cup\set{w^*} \subseteq V(\overline{G})$ is an
optimal solution to the instance $\overline{G}$ of $\ell$-round PDS
(for all $\ell\geq 4$).
\end{lemma}
\begin{proof}
The node $w^*$ should be in any optimal solution in order to power
dominate $w^*_1, w^*_2, w^*_3$, since otherwise we need to have at
least $2$ nodes from the set $\set{w^*_1, w^*_2, w^*_3}$ to get a
feasible solution. By picking $w^*$, all the nodes in $A\cup B$ (and
also the nodes inside $D_{ij}$'s that are the neighbors of $w^*$)
will be power dominated.

Assume that $A^*\cup B^*$ is an optimal solution for the
\textsc{MinRep} instance $G$. Now let us show that $\Pi=A^*\cup
B^*\cup\set{w^*}$ is a feasible solution to the PDS instance
$\overline{G}$. As described above, the nodes in $A\cup B$ and some
nodes inside $D_{ij}$'s are power dominated by $w^*$. Consider a
super edge $A_iB_j$ in ${\cal H}$. Note that the set $A^*\cup B^*$
covers all the super edges in ${\cal{H}}$. So there exists an edge
$e_q=\set{a_{i_q},b_{j_q}}\in E(G)$ such that $a_{i_q}\in A^*$ and
$b_{j_q}\in B^*$, which covers the super edge $A_iB_j$. Since
$a_{i_q}$ and $b_{j_q}$ are in the set $\Pi$ they will power
dominate their neighbors, $u_q$ and $v_q$, in all of the $4$ copies
of $D_{ij}$ in $\overline{G}$. After $u_q$ and $v_q$ are power
dominated, the node $d_q$ will power dominate the center node in
$D_{ij}$. It is easy to check that after the center node is power
dominated, all of the nodes in $D_{ij}$ will be power dominated
(through the gadgets on the nodes $\alpha, \beta, \gamma$). This
shows that $\Pi$  is a feasible solution for PDS in $\overline{G}$.
Also it is straightforward to check that such a solution will power
dominate the entire graph in at most $4$ parallel rounds. Therefore,
$\Optl{\ell}{\overline{G}}\leq \card{A^*\cup B^*}+1.$

By Proposition \ref{Pro:Monoton} the size of an optimal solution for
PDS is a lower bound on the size of an optimal solution for
$\ell$-round PDS.  So it is enough to prove that the above upper
bound is also a lower bound for PDS in $\overline{G}$. Let
$\Pi^*\subseteq V(\overline{G})$ be an optimal solution for PDS. As
we saw above, $w^*$ should be in any optimal solution for PDS. Now
define $A'=A\cap \Pi^*$ and $B'=B\cap \Pi^*$. First we prove that
any optimal solution of PDS only contains nodes from $A\cup
B\cup\set{w^*}$, and then we show that $(A',B')$ covers all the
super edges. Suppose for the contradiction that $\Pi^*$ contains
some nodes which are not in $A\cup B\cup \set{w^*}$. So there are
some $D_{ij}$'s  that cannot be power dominated completely by
$\Pi^*\cap(A\cup B\cup \set{w^*})$. By symmetry all the $4$ copies
of $D_{ij}$ are not completely power dominated.  So the optimal
solution $\Pi^*$ needs to have at least one node from at least $3$
of the $4$ copies, and the remaining one might be power dominated by
applying the propagation rule. By removing these $3$ nodes from
$\Pi^*$ and adding $a_{i_q}\in A_i$ and $b_{j_q}\in B_j$ to $\Pi^*$
for some arbitrary edge $e_q=\{a_{i_q},b_{j_q}\}\in E(G)$ we can
power dominate all of the $4$ copies of $D_{ij}$. This is a
contradiction with the optimality of $\Pi^*$. This proves that any
optimal solution will consist of nodes only from $A\cup
B\cup\set{w^*}$. To show that $(A',B')$ covers all the super edges,
it is enough to note the following: suppose no node from the inside
of any copies of $D_{ij}$ is in the optimal solution; then any
$D_{ij}$ can be power dominated only by taking both end points of an
edge between the corresponding partitions $(A_i,B_j)$. This shows
that the size of an optimal solution for PDS on $\overline{G}$ is at
least the size of an optimal solution for the \textsc{MinRep}
problem on $G$ plus $1$. This completes the proof of the lemma.
\end{proof}

\section{Integer programming formulations for $\ell$-round PDS}
In this section we present an integer programming (IP) formulation
for the $\ell$-round PDS problem, and then we present a related
integer programming formulation for the original PDS problem.
Finally, we consider LP relaxations of these two IPs, and we show
that they both have integrality gap of $\Omega(n)$.
\subsection{An integer programming for $\ell$-round PDS}
Here we consider an {IP} formulation for the $\ell$-round PDS
problem. Given an undirected graph $G=(V,E)$ and a parameter
$1\leq\ell\leq n$, where $n=\card{V}$, define the set of parallel
rounds $\T=\set{1,\ldots,\ell}$. The variables in the IP formulation
are as follows. Let $S^*$ be the optimal solution. We have a binary
variable $x_v$ for each node $v$, that is equal to $1$ if and only
if the node $v$ is in $S^*$ ($S^*=\set{v\in V: x_v=1}$). For each
node $v$ and a parallel round $t\in\T$ we have a binary variable
$z^{t}_{v}$ , where $z^t_v=1$ means that the node $v$ is power
dominated on (or before) the parallel round $t$. For each edge
$\set{u,v}\in E$ and a parallel round $t\in\T$ we have binary
variables $Y_{u \rightarrow v}^{t}$  and  $Y_{v \rightarrow u}^{t}$,
where $Y_{u \rightarrow v}^{t}=1$ means that $u$ can power dominate
$v$ at the parallel round $t+1$. Before stating the IP formulation
formally, we describe the constraints informally. We have a set of
constraints for the termination condition saying that every node
should be power dominated at the end of the last round (see $(1)$).
The second set of constraints are for power dominating all the nodes
in the closed neighborhood of a node that is in the optimal solution
(see (2)). We have another set of constraints for each edge
$\set{u,v}\in E$ that checks if the propagation rule (R2) can be
applied to $u$ in order to power dominate $v$ in the next round;
node $u$ is ready to power dominate $v$ only if $u$ and all of its
neighbors except $v$  are already power dominated (see (3)). The
last set of constraints are for checking if node $v$ is power
dominated at time $t$; node $v$ is power dominated only if either it
is in the optimal solution or at least one of its neighbors can
power dominate $v$ at time $t-1$ (see (4)). Once $v$ is power
dominated, it should remain power dominated. The term $+x_v$ in the
right hand side of constraint (4) is needed to ensure that the $z$
variables are monotone, that is, $z_v^{t+1}\geq z_v^t$.

\begin{align*}
(\texttt{IP}_{\ell})\quad &\min \sum_{v}{x_v} \quad \\
\mbox{s.t.}\quad&\\
(1)\quad & 1 \leq z^{\ell}_{v} & \forall v\in V\\
(2)\quad & z^{1}_{v} \leq  \sum_{u\in N[v]}{x_u}& \forall v\in V\\
(3)\quad & Y_{u \rightarrow v}^{t} \leq z^{t}_{w}& \forall (u,v):
\set{u,v}\in E, \forall w\in N[u]\setminus \set{v}, \forall t\in\T\\
(4)\quad & z^{t}_{v} \leq  \sum_{u\in N(v)}{Y_{u\rightarrow v}^{t-1}+x_v}& \forall v\in V, \forall t\in\T\setminus\set{1}\\
(5)\quad &\mbox{all variables are binary}\\
\end{align*}

To get an LP relaxation, we relax the variables to be non-negative
instead of being binary. It turns out that this LP relaxation is
very weak. We can add the following valid constraints that make the
LP stronger\footnote{These constraints increase the optimum value of
the LP on the cycle with $9$ nodes, $C_9$, for $\ell=3$ from $0.6$
to $1.$}. The first one forces the number of power dominated nodes
at the first round to be at least equal to the size of the smallest
closed neighborhood, and the second one forces the number of power
dominated nodes to increase by at least one at each round.
\begin{align*}
(6)\quad& \delta(G)+1\leq \sum_{v}{z^1_{v}}\\
(7)\quad& \sum_{v}{z^{t-1}_{v}}+1\leq \sum_{v}{z^{t}_v}& \hspace{3cm}\forall t\in\T\setminus\set{1}\\
\end{align*}
We now prove that the new LP relaxation has integrality gap of
$\Omega(n)$.

\begin{theorem}
Let $\ell$ be a given parameter, where $\ell=\Omega(\log{n})$. Then
the \emph{LP} relaxation for the $\ell$-round PDS problem has an
integrality gap of $\Omega(n)$ even on planar graphs.
\end{theorem}
\begin{proof}
Consider the graph $G$ that is obtained from the cycle on $m$ nodes,
$C_m$, by creating a new node $v'$  for each node $v$ and connecting
it by an edge to the original node (adding the edge $\set{v,v'}$ to
$G$) (see Figure \ref{Fig:IG-lPDS}). The graph $G$ has $n=2m$ nodes
and $2m$ edges. It follows from the proof of Theorem
\ref{Thm:NP-lPDS} that the size of a minimum power dominating set in
$G$ is at least $\lceil\frac{m}{3}\rceil$ and the optimal solution
power dominates the graph $G$ in two parallel rounds. Therefore the
size of any optimal solution for $\ell$-round PDS is at least
$\ceil{\frac{n}{6}}$ for any $\ell\geq 2$. Now we show that the LP
relaxation has  the optimum value of $O(1)$, and this completes the
proof. Let $U$ denote the set of nodes in $C_m$ in the graph $G$,
and $V$ denote the set of nodes of degree $1$ in $G$ (the newly
introduced nodes). We assign value $\alpha\geq0$ to all of the
variables corresponding to the nodes in $U$, and value $\beta\geq0$
to all of the other variables; $\forall u\in U: x_u=\alpha$ and
$\forall v\in V: x_v=\beta$. Before giving a feasible solution with
the objective value of $O(1)$ (for the LP based on $\alpha, \beta$),
we compute the value of the variables based on $\alpha$ and $\beta$
for the first few rounds.

If we apply the set of constraints (2) to all the nodes we get:
$\forall u\in U: z^1_u\leq 3\alpha+\beta$ and $\forall v\in V:
z^1_v\leq\alpha+\beta$. Let $u, v, w$ be the nodes of $G$ as shown
in Figure \ref{Fig:IG-lPDS}. By symmetry, we get 3 different
\emph{types} of $Y$ variables: $\Y{u}{v}{t}$, $\Y{v}{u}{t}$ and
$\Y{w}{u}{t}$. Again the symmetry of the graph and the constraints
imply that the value of $Y$ only depends on its type. It is easy to
check that the set of constraints (3) give the following
inequalities: $\Y{u}{v}{1}\leq 3\alpha+\beta (=z^1_u)$,
$\Y{v}{u}{1}\leq \alpha+\beta (=z^1_v)$ and
$\Y{w}{u}{1}\leq\alpha+\beta (=z^1_v)$. Next we apply constraint (4)
and we get: $z^2_u\leq
(\alpha+\beta)+2\times(\alpha+\beta)+\alpha=4\alpha+3\beta (=3z^1_v
+\alpha)$ and $z^2_v\leq 3\alpha+\beta +\beta (= z^1_u+\beta)$. We
can continue in this way, until we reach the parallel round
$t=\ell$. It is easy to prove by induction that the following
assignment for $Y$ and $z$ variables satisfy the set of constraints
(2) to (4): $\Y{u}{v}{t}=z^t_u$, $\Y{v}{u}{t}=\Y{w}{u}{t}=z^t_v$,
where $z_u$ and $z_v$ are defined recursively as follows:
\begin{align}
    &z^{t+1}_u=\alpha+3z^{t}_v \label{Equ:Zv}\tag{E1}\\
    &z^{t+1}_v=\beta+z^{t}_u \label{Equ:Zu}\tag{E2}
\end{align}
Note that  we assign the same value, $z^t_u$, to all nodes in $U$,
and the same value, $z^t_v$, to all nodes in $V$. Also we assign the
same value for all $Y$ variables of the same type. By combining the
above two recursive equations \ref{Equ:Zv} and \ref{Equ:Zu}, we get
the following independent recursive definition for $z_u$.
\begin{align}
&z^{t+2}_u=\alpha+3\beta+3z^{t}_u \tag{E3}\label{Equ:recZu}\\
&z^1_u=3\alpha+\beta\tag{E4}\\
&z^2_u=4\alpha+3\beta\tag{E5}
\end{align}
It remains to check if the above assignments also satisfy the set of
constraints (6) and (7).
\begin{align*}
&\mbox{(6): } m\cdot(3\alpha+\beta)+m\cdot(\alpha+\beta)\geq 2
\Rightarrow m\cdot(4\alpha+2\beta)\geq 2 \Rightarrow
2\alpha+\beta\geq \frac{1}{m}=\frac{2}{n}\\
&\mbox{(7): } m\cdot z^{t-1}_u+m\cdot z^{t-1}_v+1\leq m\cdot
z^{t-1}_u+m\cdot\beta+m\cdot3z^{t-1}_v+m\cdot\alpha\Rightarrow
\alpha+\beta\geq\frac{1}{m}=\frac{2}{n}
\end{align*}
Therefore, in order to satisfy the set of constraints (6) and (7),
take $\alpha, \beta\geq 0 $ be real numbers such that
$\alpha+\beta=\frac{1}{m}$. Equation \ref{Equ:recZu} for any $t\geq
1$ implies that $z^{t+2}_u\geq 3\times z^{t}_u$. By solving this
recursive inequality we get $z^{2k}_u\geq z^1_u\cdot3^{k}\geq
\frac{1}{m}\times3^{k}$. Hence, for $k=\ceil{\log_3{m}}$ we get
$z^{2k}_u\geq 1$, and consequently we get $z_v^{2k+1}\geq
z^{2k}_u\geq 1$ by Equation \ref{Equ:Zu}. Therefor by taking
$\ell=2\cdot\ceil{\log_3{m}}+1$, the set of constraints (1) is also
satisfied. Note that the above feasible solution has the objective
value of $m\times(\alpha+\beta)=1$, since $\card{U}=\card{V}=m$.
This shows that the LP has the optimum value of $O(1)$ and this
completes the proof.
\end{proof}
\begin{figure}[!h]
\begin{center}
    \input{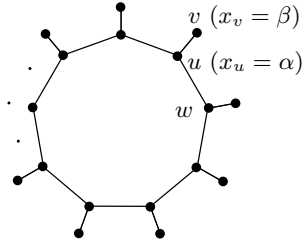}
\caption{A planar graph with large integrality gap}
\label{Fig:IG-lPDS}
\end{center}
\end{figure}

\subsection{An integer programming formulation for PDS}
As we mentioned before, the $\ell$-round PDS problem for $\ell\geq
n-1$ is the PDS problem. Here we consider an integer programming
(IP) formulation for the PDS problem that is different from the IP
formulation considered in the previous section for $\ell$-round PDS.
The IP formulation here is based on finding an ordering in which the
nodes are power dominated in the optimal solution. Given an
undirected graph $G=(V,E)$ with $n$ nodes, let
$\T=\set{1,2,\ldots,n}$. In this IP, all of the variables have the
same definition and meaning as before except $z^t_v$. Here, $z^t_v$
indicates the round in which $v$ is power dominated. The variable
$z_v^t$ is equal to $1$ if and only if $v$ is power dominated at
round $t$. Before stating the IP formulation for the PDS problem
formally, we describe the set of constraints informally. There are
two sets of constraints for presenting the ordering in which the
nodes are power dominated (see (1) and (2)). There is another set of
constraints for checking if node $u$ can power dominate node $v$ at
round $t$. Node $u$ can power dominate $v$ only if either $u$ is in
the optimal solution or all of the nodes in its closed neighborhood
except $v$ are already power dominated (see (3)). The last set of
constraints is to check if a node can be power dominated at round
$t$. Node $v$ can be power dominated at time $t$ only if either it
is in the optimal solution or at least one of its neighbors was
ready to power dominate it at the end of the previous round.
\begin{align*}
(\texttt{IP}_{O})\quad &\min \sum_{v}{x_v} \quad \\
\mbox{s.t.}\quad&\\
(1)\quad& \sum_{t}{z_v^{t}}=1 & \forall v\in V(G)\\
(2)\quad& \sum_{v}{z_v^{t}}=1 & \forall t\in\T\\
(3)\quad& Y^{t}_{u \rightarrow v}\leq \sum_{t'=1}^{t}{z_{w}^{t'}}+x_u& \forall (u,v):\set{u,v}\in E(G), \forall t\in\T\setminus\set{n}, \forall w\in N[u]-v\\
(4)\quad& z_{v}^{t}\leq \sum_{u\in N(v)}{Y^{t-1}_{u \rightarrow v}}+x_v& \forall v\in V(G), \forall t\in\T\setminus\set{1}\\
(5)\quad& \mbox{All variables are binary}\\
\end{align*}
It is easy to verify that the above IP formulates the PDS problem.
Now we consider the LP relaxation that is obtained by relaxing the
integrality of variables to nonnegativity constraints. We add the
following valid inequality to the above IP to get a stronger LP.
This new inequality forces a power domination step to occur at each
round.
\begin{align*}
(6)\quad& 1 \leq \sum_{\set{u,v}\in E}{Y^{t}_{u \rightarrow v}}& \hspace{5.5cm}\forall t\in T\setminus\set{n}\\
\end{align*}
Now we show that this LP relaxation has a big integrality gap.
\begin{theorem} The LP relaxation for PDS has an integrality gap of
$\Omega(n)$.
\end{theorem}
\begin{proof}
Consider the graph $G=(V,E)$ that is obtained from the cycle on $m$
nodes, $C_m$, by attaching a node of degree one to each node of the
cycle (see Figure \ref{Fig:IG-lPDS}). Clearly, $G$ has $n=2\times m$
nodes. It is easy to check that the size of an optimal solution for
PDS on $G$ is at least $\frac{m}{3}=\frac{n}{6}=\Omega(n)$.

Now we show that the LP relaxation has an optimum value of $O(1)$.
The graph $G$  has $n$ nodes and $n$ edges. Assign value
$\frac{1}{n}$ to each variable $x_v$ and $z_v^t$  and value
$\frac{1}{2n}$ to each variable $Y^{t}_{u \rightarrow v}$. It is
easy to check that  all of the constraints are satisfied by this
assignment. The objective value of this assignment is
$n\times\frac{1}{n}=1$, and this implies that the optimum value is
$O(1)$.
\end{proof}

\clearpage
\appendix
\renewcommand{\theequation}{A.\arabic{equation}} 
\renewcommand{\thetheorem}{A.\arabic{theorem}} 
\renewcommand{\thesection}{A}
\setcounter{equation}{0}  
\setcounter{theorem}{0}
\noindent{\Large \bf Appendix}\\
\subsection{Proofs}\label{Apx:Proof}
\begin{pf}{Proposition \ref{Thm:NP-lPDS}}
It is known that \textsc{Dominating Set} is an $\np$-hard problem on
planar graphs \cite{Bib:GAJO79}, so we just need to prove the
theorem for $\ell\geq 2$. We use (almost) the same reduction that
has been used to prove the $\np$-hardness of PDS on planar graphs
\cite{Bib:KNMRR04,Bib:GUNR05}. Given a planar graph $G=(V,E)$ and a
parameter $\ell\geq 2$ construct a graph $G'$, an instance of
$\ell$-round PDS, as follows: for each node $v\in V$ create a path
$P_v$ of length $\ell-1$ and identify an end point of $P_v$ with the
node $v$. This transformation preserves the planarity,  and the
theorem follows from the fact that the size of an optimal solution
for the \textsc{Dominating Set} problem on $G$ is equal to the size
of an optimal solution for the $\ell$-round PDS problem on $G'$. The
proof goes in the same way as the proof for the $\np$-hardness of
PDS by Kenis et al.~\cite{Bib:KNMRR04} and Guo et
al.~\cite{Bib:GUNR05}.

Assume that $S$ is an optimal solution for the \textsc{Dominating
Set} problem in $G$. Now it is easy to see that $S$ power dominates
$G'$ in exactly $\ell$ parallel rounds; the set $S$ in one parallel
round power dominates all nodes in $V$ and then in the remaining
$\ell-1$ parallel rounds each node $v\in V$ starts to power dominate
the nodes on its attached path, $P_v$.

Now assume that $S'$ is an optimal solution for the $\ell$-round PDS
problem on $G'$. It is easy to see that there is such an $S'$ that
contains only nodes of degree at least $3$, so consider that $S'$.
Assume that $S'$ is not a dominating set for $G$, so there is a node
$v\in V$ that is not power dominated in the first parallel round.
This means that $v$ is power dominated by applying the propagation
rule to one of its neighbors say $u$. It is not hard to see that
$u\in V$. Let $u'$ be the neighbor of $u$ in the path $P_u$. Since
$u$ is not in $S'$, $u'$ should also be power dominated through $u$
by applying the propagation rule. This means that $u$ power
dominates both $v$ and $u'$, but this is impossible. Therefore $S'$
should be a dominating set for $G$.
\end{pf}

In the reduction for the $\np$-hardness of PDS in
\cite{Bib:KNMRR04,Bib:GUNR05} only an edge, that is, a path of
length $1$, is attached to each node. Here,  we attach a path of
length $\ell-1$ to make an instance such that the optimal solution
needs exactly $\ell$ parallel rounds to power dominate the entire
graph.

\subsection{Dynamic programming \label{Apx:DynProg}}
Our dynamic programming is based on valid timed-orientations which
is an extension of the new formulation for PDS introduced by Guo et
al.\ \cite{Bib:GUNR05}. It is similar to the dynamic programming for
the \textsc{Dominating Set} problem given in \cite{Bib:ABFKN02} and
the dynamic programming for PDS given in \cite{Bib:GUNR05}.

Fix a parameter $\ell$ and consider the $\ell$-round PDS problem.
Assume that the graph $G=(V,E)$ and $V'\subseteq V$ and a nice tree
decomposition $\langle \set{X_i\subseteq V\ |\ i\in
I},T=(I,F)\rangle$ of $G$ with tree-width $k$ are given as input.
Let $T_i$ denote the subtree of $T$ rooted at node $i\in I$, and
$Y_i$ denote the set $(\bigcup_{j\in V(T_i)}{X_j})\setminus X_i$.
Furthermore, let $G_i$ be the subgraph induced on $Y_i\cup X_i$,
i.e. $G_i=G[Y_i\cup X_i]$. Also denote by $G'_i$ the subgraph
induced on $X_i$.  Let $n_i$ and $m_i$ be the number of nodes and
the number of edges in $G'_i$ respectively. The dynamic programming
works on the bottom-up fashion. On each bag $X_i$, it considers all
 valid timed-orientations of the subgraph $G_i$ and stores the
number of origins together with the orientation on the edges of
$X_i$ as the states of the bag $X_i$. In the other words, the valid
timed-orientation is stored through the states of the bag.

{\noindent \bf The state of a bag:} The state $s$ for a bag $X_i$
defines the orientation of the edges inside $G'_i$, the time label
assigned to the nodes of $X_i$, the number of directed edges from
$v\in X_i$ to all nodes in $Y_i$, and also the maximum of the
time-label assigned to the neighbors of  $v$ in the set $Y_i$. In a
bag state $s$ we denote the state of an edge $e=\set{u, v}\in
E(G'_i)$ by $s(e)$, the time-label assigned to $v\in X_i$ by
$s_t(v)$, the number of incoming edges from $Y_i$ to $v$ by
$s_{-}(v)$, the number of outgoing edges from $v$ to $Y_i$ by
$s_{+}(v)$ and the maximum of the time-label assigned to the
neighbors of $v$ in $Y_i$ by $s_y(v)$. Let $e=\set{u,v}$ be an edge
in $G'_i$, then $s(e)$ takes one of the following $3$ values:
``$u\rightarrow v$'', ``$v\rightarrow u$'' or ``$\bot$''; where the
first two values shows the direction of the edge $e$ in the valid
timed-orientation and the third one indicates that $e$ is left
undirected. Consider a node $v\in X_i$, $s_t(v)$ takes a value from
$\set{0}\cup\set{1, 2, \ldots, \ell}\cup\set{\hat{1}, \hat{2},
\ldots, \hat{\ell}}\cup\set{+\infty}$.  The node $v$ with
$s(v)=\hat{a}$ means that we still ask for a propagation rule to be
applied to a neighbor of $v$ and power dominate it, and $s(v)=a$
shows that the node $v$ is already power dominated at the current
stage of the algorithm. We use  $\uh{s(v)}$ to denote the integer
value of the label $s(v)$ ignoring the hat notation (e.g.
$\uh{\hat{2}}=2$). Also $s_{-}(v)$ takes a value from $\set{0, 1}$,
 $s_{+}(v)$ takes a value from $\set{0, 1, 2}$ and $s_y(v)$ takes
a value from $\set{0, 1, \ldots, \ell}\cup\set{+\infty}$. Values of
$0$ or $1$ for $s_{-}(v)$ and $s_{+}(v)$ shows the exact number of
incoming/outgoing edges to/from $v$, but $2$ means that there are at
least $2$ edges. Let us denote by $\mathcal{S}_i$ the set of all
possible states for the bag $X_i$. It is straightforward to check
that the number of bag states for $X_i$ is
$\card{\mathcal{S}_i}=3^{m_i}\times(2\ell+2)^{n_i}\times
5^{n_i}\times(\ell+2)^{n_i}$. Note that a node cannot have
$s_{-}(v)=1$ and $s_{+}(v)=2$ at the same time (this follows easily
from Definition \ref{Def:ValTOR}), so $s_{-}(v)$ and $s_{+}(v)$ have
$5$ different combinations.

For each bag $X_i$ we will compute and store a mapping
$A_i:\mathcal{S}_i\rightarrow \mathbb{N}\cup\set{+\infty}$. For a
bag state $s\in\mathcal{S}_i$, the value $A_i(s)$ shows the minimum
number of origins in the optimal valid timed-orientation of the
subproblem induced on $G_i$ under the restriction that the
orientation of edges and labeling of nodes in $X_i$ is defined by
the state $s$. A bag state $s\in\mathcal{S}_i$ for the bag $X_i$ is
called \emph{invalid} if
\begin{align*}
\begin{split}
(\neg {\mbox{P1}}\equiv)&\left(\exists v\in V'\cap X_i: s_t(v)=+\infty\right) \vee\\
(\neg {\mbox{P2}}\equiv)&\left(\exists v\in X_i:(1\leq \uh{s_t(v)}\leq\ell)\wedge (d_i^{-}(v)+s_{-}(v)> 1)\right)\vee\\
(\neg {\mbox{P3}}\equiv)&\left(\exists v\in X_i: s_t(v)=+\infty\wedge (d_i^{-}(v)+s_{-}(v)+d_i^{+}(v)+s_{+}(v)\geq 1)\right)\vee\\
(\neg {\mbox{P4}}\equiv)&\left(\exists v\in X_i:
s_t(v)=0\wedge d_i^{-}(v)+s_{-}(v)\geq 1\right)\vee\\
(\neg {\mbox{P5}}\equiv)&\left(\exists e=\set{u,v}\in E(G'_i): s(e)=``u\rightarrow v" \wedge ((\uh{s_t(v)}=1\wedge \uh{s_t(u)}\neq 0)\vee(\uh{s_t(v)}>1\right.\\
&\hspace{2cm}\left.\wedge \uh{s_t(v)}< 1+\max\set{s_y(u)}\cup\set{\uh{s(w)}: w\in N_i[u]-v}))\right)\vee\\
&\left(\exists v\in X_i, \exists a\in\set{\hat{1}, \ldots,
\hat{\ell}}: (s_t(v)=a \wedge d_i^{-}(v)+s_{-}(v)\neq0)\right)\vee\\
&\left(\exists a\in\set{1, \ldots, \ell}: s_t(v)=a \wedge
d_i^{-}(v)+s_{-}(v)=0\right)
\end{split}
\end{align*}
where $d_i^{-}(u)$, $d_i^+(u)$, and $N_i[u]$ denote respectively the
in-degree, out-degree, and closed neighborhood of $u$ in the graph
 that is obtained from $G'_i$ by orienting edges according to the
 state $s$. Recall that the in-degree/out-degree shows the
 number of directed incoming/outgoing edges, but for the closed
 neighborhood we consider the undirected graph $G'_i$.
 Informally a bag state is invalid if it either violates any
one of the valid timed-orientation's properties (P1 to P5), or it
cannot be extended to a valid timed-orientation.
Now we describe our dynamic programming:

{\noindent \bf Step 1: (Initialization)} In this step for each leaf
node $i$ in the tree $T$, we define (initialize) the mapping $A_i$
for each $s\in\mathcal{S}_i$ as follows:
$$A_i(s)=\left\{\begin{array}{ll}
+\infty & \mbox{if either $s$ is  invalid or $(\exists v\in X_i: s_{-}(v)+s_{+}(v)+s_y(v)\neq 0)$}\\
\card{\set{v\in X_i: s(v)=0}}& \emph{otherwise}\\
\end{array}\right.$$

{\noindent \bf Step 2: (Bottom-Up Computation)} In this step we
compute in the bottom-up fashion from leaves to root the mapping
corresponding to each bag in the tree. Recall that the tree nodes
have three types: \textsc{Join Node, Insert Node, Forget Node}. In
the following we describe how to compute $A_i$ in each of these
three cases. In each point of the algorithm we preserve the
following invariant: for each $s\in \mathcal{S}_i$ there exists a
valid timed-orientation for $G_i$ which is compatible with $s$ and
has minimum number of origins equal to $A_i(s)$ where the power
domination of all nodes in $G_i$ are justified except for the ones
in $\set{v\in X_i: s(v)=\hat{a}\mbox{ for some }a\in\set{1, \ldots,
\ell}}$.

{\noindent \bf \textsc{Forget Node:}} Suppose $i$ is a forget node
with child $j$, and assume that $X_j=X_i\cup\set{x}$. The bag states
$s\in\mathcal{S}_i$ and $s'\in\mathcal{S}_j$ are called
\emph{forget-compatible} and denoted by  $s\overset{F}{\sim} s'$, if
\begin{itemize}
\item[(F1)] $\forall e\in E(G'_j): s(e)=s'(e)$,
\item[(F2)] $\forall v\in V(G'_j): s_t(v)=s'_t(v)$
\item[(F3)] $\forall v\in V(G'_j): s_{-}(v)=s'_{-}(v)+ [s'(\set{x, v})=``x \rightarrow v"] \wedge s_{+}(v)=s'_{+}(v)+
[s'(\set{x, v})=``v \rightarrow x"]$
\item[(F4)] $s'(x)\in \set{0, 1, \ldots, \ell}\cup\set{+\infty}$
\item[(F5)] $\forall v\in V(G'_j): s_y(v)=\begin{cases}
\max\{s'_y(v), s'_t(x)\} & \mbox{if $\set{x, v}\in E(G'_j)$}\\
s'_y(v) &\mbox{otherwise}
\end{cases}$,
\end{itemize}
where $[P]$ is equal to $1$ if $P$ is a true statement and $0$
otherwise. Now we compute the mapping $A_i$ for the bag $X_i$ as
follows: $\forall s\in \mathcal{S}_i$
$$A_i(s)=\begin{cases}
+\infty & \mbox{if $s$ is invalid}\\
\min\{A_j(s'): s'\in \mathcal{S}_j,s\overset{F}{\sim}s'\} &
\mbox{otherwise} \end{cases}$$ Note that since $x$ is not in $X_i$,
by property (3) of the tree decomposition (Definition \ref{Def:TWD})
it will never appears in any bag in the rest of the algorithm. This
implies that the power domination of $x$ should be justified within
$G_j$.

{\noindent \bf\textsc{Insert Node:}} Suppose $i$ is an insert node
with child $j$, and assume that $X_i=X_j\cup\set{x}$.
We introduce a mapping $\phi:\mathcal{S}_i\rightarrow\mathcal{S}_j$
in the following way. We map a given state $s\in\mathcal{S}_i$ to
$s'=\phi(s)\in\mathcal{S}_j$ in the following way.
\begin{itemize}
\item[(I1)] $\forall e\in E(G'_j): s'(e)=s(e)$
\item[(I2)] $\forall v\in V(G'_j): s'_t(v)=
\begin{cases}
\hat{a} & \mbox{if $s_t(v)=a$ and $[s(\set{x, v})=``x \rightarrow v"]=1$}\\
s_t(v) & \mbox{otherwise}
\end{cases}$
\item[(I3)] $\forall v\in V(G'_j): s'_{-}(v)=s_{-}(v), s'_{+}(v)=s_{+}(v), s'_y(v)=s_y(v)$
\end{itemize}
We now compute the mapping $A_i$ for the bag $X_i$ as follows:
$\forall s\in \mathcal{S}_i$
$$A_i(s)=
\begin{cases}
+\infty & \mbox{if $s$ is  invalid or $s_{-}(x)+s_{+}(x)+s_y(x)\neq 0$}\\
A_j(\phi(s))+[s(x)=0]& \mbox{otherwise}
\end{cases}$$
Again note that, since $x$ appears in $X_i$ but not in $X_j$ by
property (3) of tree decomposition it is the first time that it
appears in the subtree rooted at node $i$. Also note that $x$ cannot
have a neighbor in $Y_i$.

{\noindent \bf\textsc{Join Node:}} Suppose $i$ is a join node with
$j$ and $k$ as its children, and assume that $X_i=X_j=X_k$. We say
$s'\in \mathcal{S}_j$ and $s''\in \mathcal{S}_k$ are
\emph{join-compatible} with $s\in\mathcal{S}_i$ and denote it by
$s\overset{J}{\sim}(s',s'')$, if the following conditions hold:
\begin{itemize}
\item[(J1)] $\forall v\in V(G'_i): (s_{-}(v)=s'_{-}(v)+s''_{-}(v)) \wedge (s_{+}(v)=s'_{+}(v)+s''_{+}(v)) \wedge
(s_y(v)=\max\{s'_y(v), s''_y(v)\})$
\item[(J2)] $\forall v\in X_i: s_t(v)\in\set{0}\cup\set{\hat{1}, \ldots, \hat{\ell}}\cup\set{+\infty} \Rightarrow
s_t(v)=s'_t(v)=s''_t(v)$,
\item[(J3)] $\forall v\in X_i: s_t(v)=a\in\set{1, \ldots, \ell} \Rightarrow \left( \left(\uh{s'_t(v)}=\uh{s''_t(v)}=a\right)
\wedge \left( s'_t(v)=a \vee s''_t(v)=a\right)\right)$
\item[(J4)] $\forall e\in E(G'_i): (s(e)=``\bot" \Rightarrow
s'(e)=s''(e)=``\bot") \wedge \\( s(e)=``u\rightarrow v" \Rightarrow
(s'(e)\neq ``v\rightarrow u" \wedge s''(e)\neq ``v\rightarrow u"))$
\end{itemize}
Informally  $s$ is join-compatible with the pair of states
$(s',s'')$ if they are assigning the same label to the nodes and the
same orientation to the edges, and also if a node is power dominated
it should be justified in either $s'$ or $s''$. We now compute the
mapping $A_i$ for the bag $X_i$ as follows: $\forall s\in
\mathcal{S}_i$
$$A_i(s)=
\begin{cases}
+\infty & \mbox{$s$ is  invalid}\\
\min{\set{A_{j}(s')+A_{k}(s'') -\card{\set{v\in X_i: s(v)=0}}:
\mbox{$s'\in\mathcal{S}_j$, $s''\in\mathcal{S}_k$,
$s\overset{J}{\sim}(s',s'')$}}} & \mbox{otherwise}
\end{cases}
$$

{\noindent \bf Step 3: (At root $r$)} Let $r$ be the root of the
tree decomposition $T$. Finally we compute the number of origins in
the optimal solution for $\ell$-round PDS on the instance $\langle
G,V'\rangle$ in the following way:
$$\min{\set{A_r(s): s\in\mathcal{S}_r, \forall v\in X_r: s(v)\in\set{0, 1, \ldots,
\ell}\cup\set{+\infty}}}.$$

Recall  that the number of bag states is
$\card{\mathcal{S}_i}=3^{m_i}\times(2\ell+2)^{n_i}\times
5^{n_i}\times(\ell+2)^{n_i}$. Let $m_e$ be the maximum number of
edges that a bag can have, and also note that the maximum number of
nodes is $(k+1)$. It is easy to check that each step of the dynamic
programming can be computed in time $O(c^{m_e+k\log{\ell}})$, for
some global constant $c$. This shows that the total running time of
our algorithm is $O(c^{m_e+k\log{\ell}}\cdot \card{V})$.

\end{document}